\documentclass[letterpaper, 10 pt, conference]{ieeeconf}
\IEEEoverridecommandlockouts
\usepackage{hyperref}
\hypersetup{
colorlinks=true,
linkcolor=blue,
anchorcolor=blue,
citecolor=blue}
\usepackage{caption} 
\usepackage{subcaption}
\usepackage{graphicx}
\usepackage{epstopdf}
\usepackage{cite}
\usepackage{amsmath,amssymb,amsfonts}
\usepackage{algorithmic}
\usepackage{textcomp}
\usepackage{xcolor}
\usepackage{booktabs}
\usepackage[ruled,linesnumbered]{algorithm2e}
\usepackage{makecell}
\newtheorem{theorem}{\it Theorem}
\newtheorem{remark}{\it Remark}
\newtheorem{definition}{\it Definition}
\newtheorem{assumption}{\it Assumption}
\def\BibTeX{{\rm B\kern-.05em{\sc i\kern-.025em b}\kern-.08em
    T\kern-.1667em\lower.7ex\hbox{E}\kern-.125emX}}
\begin{document}

\title{\LARGE \bf Observer-Based Distributed Model Predictive Control for String-Stable Multi-vehicle Systems with Markovian Switching Topology

\author{Wenwei Que$^{1}$, Yang Li$^{1}$, Lu Wang$^{1}$, Wentao Liu$^{1}$, Yougang Bian$^{1}$, Manjiang Hu$^{1}$, Yongfu Li$^{2}$}

\thanks{This work was supported by the National Key R\&D Program of China under Grant 2023YFB2504701 and 2023YFB2504704, in part by the Natural Science Foundation of Chongqing under Grant No. CSTB2022NSCQ-LZX0025, the Talent Program of Chongqing under Grant No. cstc2024ycjh-bgzxm0037, the Science and Technology Research Program of Chongqing Municipal Education Commission under Grant No. KJZD-M202300602. (\textit{Corresponding Author: Yang Li})}

\thanks{$^{1}$ Wenwei Que, Yang Li, Lu Wang, Wentao Liu, Yougang Bian, and Manjiang Hu are with the College of Mechanical and Vehicle Engineering, Hunan University, Changsha 410082, China. (email:quewenwei@hnu.edu.cn; lyxc56@gmail.com; wanglu\_aine@163.com; 3187206021@qq.com; byg10@foxmail.com; manjiang\_h@hnu.edu.cn)}
\thanks{$^{2}$
Yongfu Li is with the Key Laboratory of Intelligent Air-Ground Cooperative Control for Universities in Chongqing, School of Automation, Chongqing University of Posts and Telecommunications, Chongqing 400065, China (e-mail: liyongfu@ieee.org)}
}

\maketitle
\thispagestyle{empty}
\pagestyle{empty}
\begin{abstract}
Switching communication topologies can cause instability in vehicle platoons, as vehicle information may be lost during the dynamic switching process. This highlights the need to design a controller capable of maintaining the stability of vehicle platoons under dynamically changing topologies. However, capturing the dynamic characteristics of switching topologies and obtaining complete vehicle information for controller design while ensuring stability remains a significant challenge. In this study, we propose an observer-based distributed model predictive control (DMPC) method for vehicle platoons under directed Markovian switching topologies. Considering the stochastic nature of the switching topologies, we model the directed switching communication topologies using a continuous-time Markov chain. To obtain the leader vehicle's information for controller design, we develop a fully distributed adaptive observer that can quickly adapt to the randomly switching topologies, ensuring that the observed information is not affected by the dynamic topology switches. Additionally, a sufficient condition is derived to guarantee the mean-square stability of the observer. Furthermore, we construct the DMPC terminal update law based on the observer and formulate a string stability constraint based on the observed information. Numerical simulations demonstrate that our method can reduce tracking errors while ensuring string stability.
\end{abstract}

\section{Introduction}\label{Introduction}
\par Connected and automated vehicles (CAVs) play a crucial role in intelligent transportation systems (ITS). 
Vehicle platooning is one of the key applications of CAVs, with significant potential to improve traffic safety, alleviate congestion, enhance fuel efficiency, and reduce emissions \cite{braiteh2024platooning, Zhu2024Finite, Feng2024, 10774505}. In vehicle platooning, vehicle-to-vehicle (V2V) communication technology is used for inter-vehicle information exchange, which is typically described by a communication topology\cite{Wen2024Adaptive}.
In practice, the communication topology may be switched or time-varying.  The reasons for switching topologies include the limitation of wireless communication range, the complex traffic, and unreliable communication conditions (e.g., time delays, packet loss, signal blocking, or communication failures). Frequent switching of the communication topology can lead to the loss of global information and increased communication delays, making it difficult to ensure the stability of vehicle platooning control and potentially leading to collisions \cite{li2020distributed}. This requires a safe and stable (including closed-loop stability and string stability) platoon controller that can work efficiently under highly dynamic and switching communication topologies. However, capturing the highly dynamic characteristics of switching topologies and accurately obtaining complete vehicle information during the topology switching process for controller design remains a challenging task.

\begin{figure}[!t]
    \centering
    \begin{subfigure}{0.48\linewidth}
        \includegraphics[width=\linewidth]{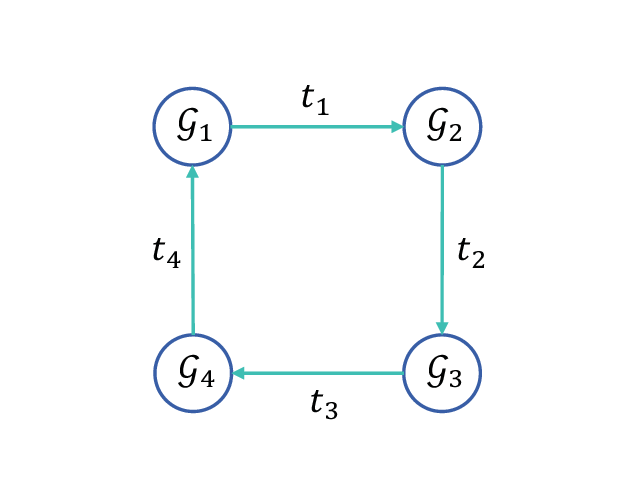}
        \caption{predefined switching topology}
        \label{fig:subfig_predefined_switching}
    \end{subfigure}
    \hfill
    \begin{subfigure}{0.48\linewidth}
        \includegraphics[width=\linewidth]{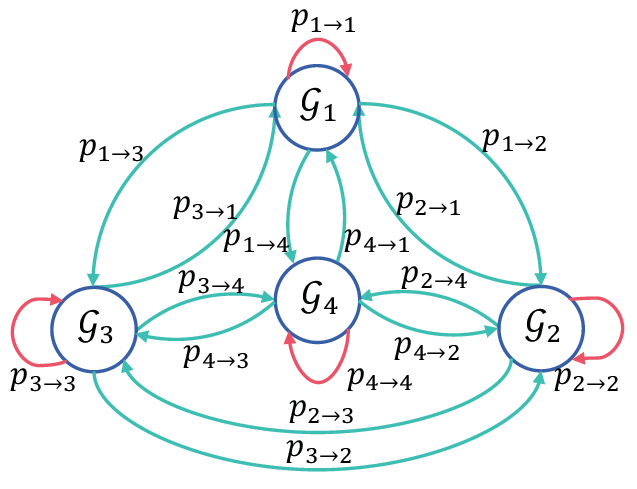}
        \caption{stochastic switching topology}
        \label{fig:subfig_stochastic_switching}
    \end{subfigure}
    \caption{Two switching strategies: (a) predefined switching topology, switching between different graphs with a preset dwell time. (b) stochastic switching topology, switching with a random probability.}
    \label{fig:switching_mainfig}
    \vspace{-0.6cm}
\end{figure}

\par The controller design with switching communication topologies depends on the way of switching, as shown in Fig. \ref{fig:switching_mainfig}, the left involves predefined switching strategies \cite{li2020distributed, Wang2022model, Jun2023Distributed} and the right considers stochastic switching \cite{Ding2024Distributed, wang2022fully, wang2024cyber, Wen2024Adaptive}. Predefined switching strategies usually assume a lower switching frequency, and thus cannot effectively simulate real-world random and dynamic communication environments \cite{li2020distributed}. In contrast, stochastic switching topologies can better simulate the practical communication characteristics (such as Rayleigh fading channel \cite{wang2024cyber}), which have recently become a research focus. For instance, \cite{wang2022fully} uses the continuous-time Markov processes to characterize the time-varying communication topologies, and proposes a fully distributed adaptive anti-windup control protocol for CAVs platooning with switching topologies and input saturation, ensuring closed-loop stability. However, string stability is not considered in \cite{wang2022fully}. \cite{Ding2024Distributed} also investigates a distributed adaptive control scheme with Markovian switching communication topologies. However, it requires global information (i.e., a directed spanning tree rooted at the leader) for all communication topologies. \cite{wang2024cyber} design a fully distributed adaptive observer against the switching topologies on the control stability. However, due to the limitations of the observer design, it only works for the undirected switch topologies. Overall, existing studies on the stability of platoon control with switching communication topologies either overlook string stability or have a strong dependency on the way of switching topology.
\begin{figure}[!t]
\includegraphics[width=0.95\linewidth]{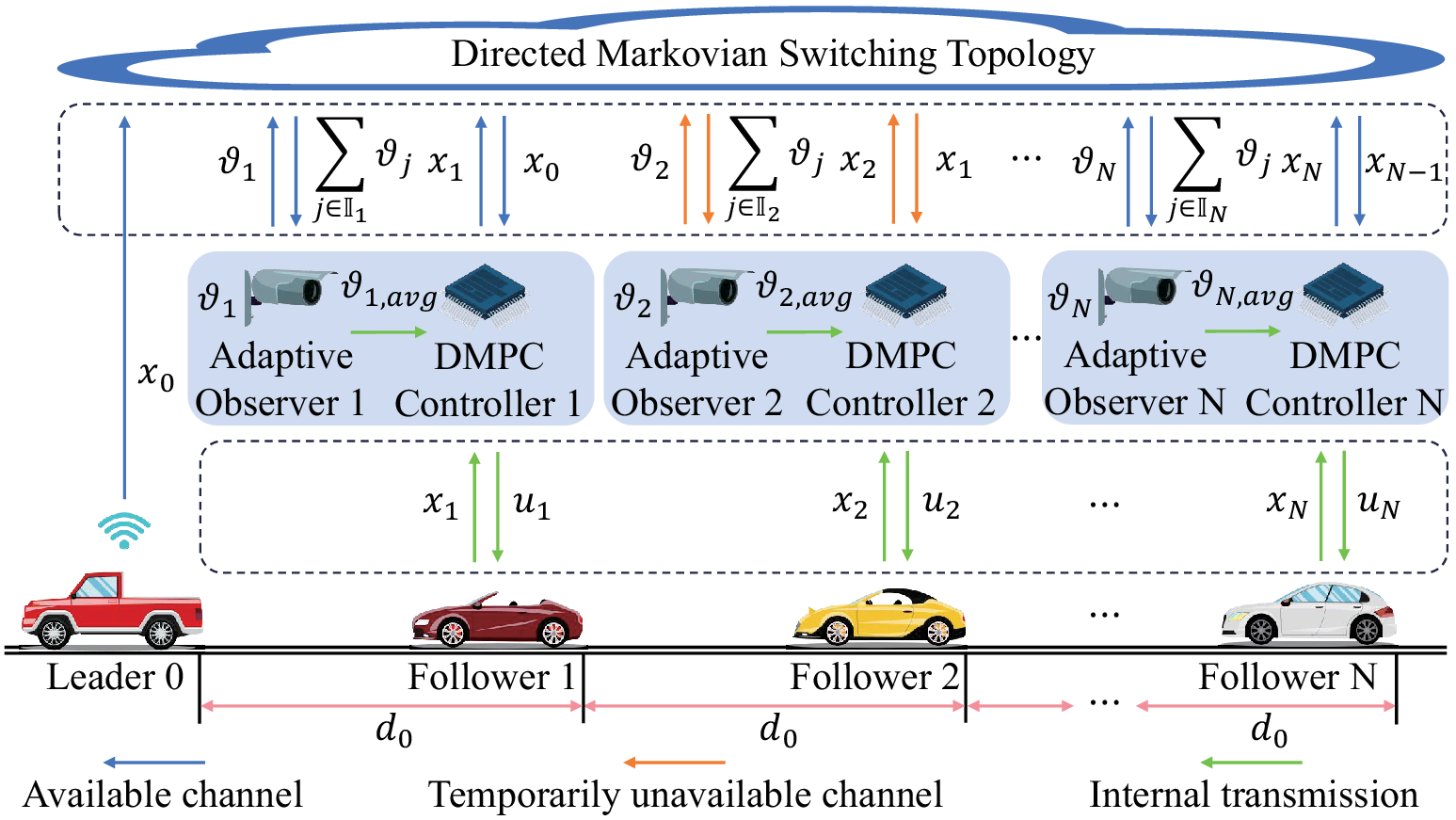}
\centering
\caption{Framework of the proposed DMPC controller. The adaptive observer receives the neighbor's observation via V2V communication, averages it with its own observation, and transmits the result to the DMPC controller. $\vartheta_{i}$ is the state of observer $i$; $\mathbb{I}_i$ is the in-neighbor set of vehicle $i$ ;$\vartheta_{i,avg}$ is the average observation of vehicle $i$. $x_i$ and $u_i$ are the state and control of each vehicle $i$.
The DMPC controller monitors its own state, receives the averaged observation $\vartheta_{i,avg}$ via internal transmission, and obtains the preceding vehicle's state via V2V. }
\label{fig:framework}
\vspace{-0.6cm}
\end{figure}
\par This paper investigates an observer-based distributed model predictive control (DMPC) controller with directed Markovian switching communication topology, as shown in Fig. \ref{fig:framework}. We developed a fully distributed adaptive observer that can quickly adapt to randomly switching topologies. To ensure stability, we designed a string-stable DMPC controller for switching topologies based on the observation of the leader vehicle.
\begin{itemize}
\item  A fully distributed adaptive observer is designed to acquire the state of the leader vehicle under the directed Markovian switching communication topology. This observer does not rely on global information and can quickly adapt to the randomly switching topologies, ensuring that the observed information remains unaffected by the switching of communication topologies. 

\item A distributed model predictive control framework is proposed, in which the terminal update law is constructed based on the observer to guarantee
asymptotic terminal consensus. The string stability constraints are also built using the average observed state of the leader vehicle, ensuring the predecessor-follower string stability.

\item A sufficient condition for the asymptotic mean-square stability of the observer error system is derived using the algebraic Riccati equation. Numerical simulation results demonstrate that the proposed method can achieve faster convergence, smaller tracking errors, and string stability under Markovian switching communication.
\end{itemize} 
\par The rest of the paper is organized as follows: Section~\ref{Preliminaries and Problem Statement} provides the preliminaries and problem statement, and Section~\ref{Control Algorithm Design} introduces the observer-based DMPC method. The numerical simulation results are shown in Section~\ref{Numerical Experiments}, and Section~\ref{Conclusion} concludes the paper.

\section{Preliminaries and Problem Statement}\label{Preliminaries and Problem Statement}

\subsection{Switching Communication Topology Modeling}
We consider a vehicle platoon system consisting of a leader and $N$ followers, with communication topology among them represented by a directed graph $\mathcal{G}\triangleq(\mathcal{V},\mathcal{E})$, where $\mathcal{V}=\{0,1,2,\cdotp\cdotp\cdotp,N\}$ is the set of nodes, representing the collection of vehicles within the platoon. The followers are $i,j\in\mathbb{N}=\{1,2,\cdotp\cdotp\cdotp,N\}$ and the leader is $0$. $\mathcal{E}\subseteq\mathcal{V}\times\mathcal{V}$ is the set of edges, representing the directed connectivity between CAVs in the platoon. Additionally, if there exists a tree-type subgraph that includes all nodes of $\mathcal{G}$ and has the leader as its initial node, it is called a directed spanning tree rooted at the leader. We define the adjacency matrix $\mathcal{A}=\begin{bmatrix}a_{ij}\end{bmatrix}\in\mathbb{R}^{N\times N}$, when vehicle $i$ can receive information from vehicle $j$, $a_{ij}=1$, otherwise, $a_{ij}=0$. Here we assume there is no self-loop, i.e., $a_{ii}=0$. Define the spinning matrix $\mathcal{W}=\text{diag}\begin{bmatrix}\omega_{i0}\end{bmatrix}\in\mathbb{R}^{N\times N}$, when vehicle $i$ can receive information from the leader, $\omega_{i0}=1$, otherwise, $\omega_{i0}=0$. Define the Laplacian matrix $\mathcal{L}=\begin{bmatrix}l_{ij}\end{bmatrix}\in\mathbb{R}^{N\times N}$ , where $l_{ii}=\sum_{j\neq i}^Na_{ij}$ and  $l_{ij}=-a_{ij},i\neq j$. Define the information matrix $\mathcal{M}=\begin{bmatrix}m_{ij}\end{bmatrix}\in\mathbb{R}^{N\times N}$ and $\mathcal{M}=\mathcal{W}+\mathcal{L}$. The in- and out-neighbor sets of vehicle $i$ are respectively defined as $\mathbb{I}_i=\{j\in\mathbb{N},j\neq i|a_{ij}=1\}$ and $\mathbb{O}_i=\{j\in\mathbb{N},j\neq i|a_{ji}=1$\}.
We use time-varying switching graphs $\mathcal{G}(\sigma(t))\in\{\mathcal{G}_{1},\mathcal{G}_{2},\cdotp\cdotp\cdotp,\mathcal{G}_{\iota}\}$ to represent all possible communication topologies. If and only if $\sigma(t)=\varphi \in\{1,2,\cdots,\iota\}$, $\mathcal{G}(\sigma(t))=\mathcal{G}_{\varphi }$, $\iota$ is the number of topologies. An infinitesimal generator is defined by the transition rate matrix  $\mu=(\mu_{\varphi q})\in\mathbb{R}^{\iota\times\iota}$ \cite{meyn2012markov}, for any positive scalar $\Delta t$, as $\Delta t\to0$, one has
\begin{align}\label{eq:infinitesimal generator}
\left.\begin{aligned}&\mathbb{P}\left\{\sigma(t+\Delta t)=q\mid\sigma(t)=\varphi \right\}\\&=\left\{\begin{array}{l}\mu_{\varphi q}\Delta t +o(\Delta t),\text{when }\sigma(t)\text{ jumps from }\varphi \text{ to }q,\\1+\mu_{\varphi \varphi }\Delta t+o(\Delta t),\text{ otherwise},\end{array}\right.\\\end{aligned}\right.
\end{align}
where $\mu_{\varphi q}$ is the transition rate from state $\varphi $ to state $q$ , if $\varphi =q$ , $\mu_{\varphi \varphi }=-\sum_{\varphi \neq q}\mu_{\varphi q}$, otherwise, $\mu_{\varphi q}\geq0$, $o(\Delta t)$ stands for an infinitesimal of a higher order than $\Delta t$, satisfying $\lim_{\Delta t\to0}\frac{o(\Delta t)}{\Delta t}=0$.
\begin{assumption}\label{ass:continuous-time Markov process}
The switching process of the communication topology is controlled by a continuous-time Markov process $\{\sigma(t),t\geq0\}$ whose transition rate matrix $\mu$ is ergodic.
\end{assumption}
\begin{remark}
    The ergodic Markov process has a unique invariant distribution $\pi=\begin{bmatrix}\pi_{1},\pi_{2},\cdotp\cdotp\cdotp,\pi_{\iota}\end{bmatrix}^{T}$ fulfilling $\sum_{q=1}^{\iota}\pi_{q}=1$ and  $ \pi_{q}\geq0$.
\end{remark}
\begin{assumption}\label{ass: a directed spanning tree rooted at the leader}
Define a union graph of switching graphs $\bar{\mathcal{G}}\triangleq\begin{aligned}\bigcup_{\varphi =1}^{\iota}\mathcal{G}_{\varphi}=\left(\mathcal{V},\bigcup_{\varphi=1}^{\iota}\mathcal{E}_{\varphi}\right)\end{aligned}$, $\bar{\mathcal{G}}$ contains a directed spanning tree rooted at the leader. 
\end{assumption}
\par We use $\mathcal{L}(\sigma(t))$, $\mathcal{M}(\sigma(t))$, $\mathbb{I}_i(\sigma(t))$, and  $\mathbb{O}_i(\sigma(t))$ to indicate the time-varying Laplacian matrix, information matrix, in-neighbor sets, and out-neighbor set corresponding to $\bar{\mathcal{G}}(\sigma(t))$. Based on Assumption \ref{ass: a directed spanning tree rooted at the leader}, $\mathcal{M}(\sigma(t))$ and its mathematical expectation $\mathbb{E}[\mathcal{M}(\sigma(t)]$ are both positive definite.
\begin{remark}
Different from \cite{wang2024cyber},  we consider directed switching graphs with asymmetric information matrix $\mathcal{M}(\sigma(t))$. Similar to \cite{wang2022fully}, the convergence analysis of the observer is challenging.    
\end{remark}
\subsection{Vehicle Longitudinal Dynamics}
\par We consider the longitudinal dynamics of a vehicle platoon on a straight and flat road. To simplify the problem,  the nonlinear third-order model for CAVs is formulated as follows:
\begin{align}\label{eq:continuous follower model}
\begin{cases}\dot{p}_i(t)=v_i(t),\\\dot{v}_i(t)=a_i(t),
\\\frac{\eta_{i}}{r_{w,i}}T_i(t)=m_ia_i(t)+C_{A,i}v_i^2(t)+m_igf_i,
\\\delta_i\dot{T}_i(t)+T_i(t)=T_{des,i}(t),i\in\{1, 2,\cdots, N \},\end{cases}
\end{align}
where $p_{i}(t)$, $v_{i}(t)$, $a_{i}(t)$  are the position, velocity, and acceleration of vehicle $i$, respectively. $m_i$, $\eta_{i}$, $r_{w,i}$, $C_{A,i}$, $g$, $f_i$ are vehicle mass, mechanical efficiency, tire radius, aerodynamic drag, gravitational acceleration, and rolling resistance, respectively. $\delta_i$ represents the inertial time lag coefficient. $T_i(t)$ and $T_{des,i}(t)$ are actual and desired control torque on the wheels. We discretize the follower vehicle's dynamics model by the precise feedback linearization strategy,
 \begin{align}\label{eq:feedback linearization strategy}
T_{des,i}(t)=&\frac{r_{w,i}}{\eta_{i}}\nonumber\bigg(m_iu_i(t)+m_igf_i\\&+C_{A,i}v_i(t)\Big(2\delta_i\dot{v}_i(t)+v_i(t)\Big)\bigg),
 \end{align}
where $u_i(t)$ is the control input. Referring to \cite{li2020distributed}, the linear third-order model for follower $i$ is given as follows:
\begin{align}\label{eq:continuous linear follower model}
     &\dot{x}_i(t)=Ax_i(t)+Bu_i(t),\nonumber\\
     &A=\begin{bmatrix}0&1&0\\0&0&1\\0&0&0\end{bmatrix},B=\begin{bmatrix}0\\0\\1\end{bmatrix},
\end{align}
where $x_i(t)=[p_i(t),v_i(t),a_i(t)]^T$ is the state of follower $i$.
\par The leader vehicle follows the reference trajectory of the virtual leader, which can be represented using the following linear system:
\begin{align}\label{eq:continuous leader model}
\dot{x}_{0}(t)=Ax_{0}(t), 
\end{align}
where $x_0(t)=[p_{0}(t),v_{0}(t),a_{0}(t)]^{T}$ is the state of the leader.
\par This study does not assume that the leader's acceleration is zero, and the velocity of the leader might change over time. Considering the time-varying topology, the state of the virtual leader is not globally available and can only be received by a subset of followers through V2V communication. 

\subsection{Control Objective}
\par  We denote the tracking error for follower $i$ as $e_i(t)=x_i(t)-x_0(t)+\tilde{d}_{i0}=[e_{p,i}(t),e_{v,i}(t),e_{a,i}(t)]^T$, where $\tilde{d}_{i0}=[i\cdotp d_0,0,0]^T$, $d_{0}$ is the desired constant distance. $e_{p,i}(t)$, $e_{v,i}(t)$, and $e_{a,i}(t)$ are the position, velocity, and acceleration errors at time $t$, respectively. The platoon system \eqref{eq:continuous linear follower model} with Markovian switching topology is required to achieve mean-square closed-loop stability and string stability in this paper, defined as follows: 
\begin{definition}[Mean-square closed-loop stability\cite{bian2024distributed}]
    System \eqref{eq:continuous linear follower model} satisfies mean-square closed-loop stability if the following condition holds:
\begin{align}\label{eq:Mean-square closed-loop stability}
\lim_{t\to+\infty}\mathbb{E}\left[\|e_i(t)\|^2\right]=0,
\end{align}
where $\mathbb{E}[\cdot]$ represents the mathematical expectation, and $\|\cdot\| $ represents the Euclidean norm.
\end{definition}
\begin{definition}[Predecessor-follower string stability\cite{dunbar2011distributed}]
For a step change of $v_0(0)$, position errors of  follower $i=2,\cdots,N$ can converge to 0 and there exist a constant $\beta_i\in(0,1)$ such that position error satisfy
\begin{align}\label{eq:Predecessor-follower string stability}
    \|e_{p,i}(t)\|_\infty\leq \beta_i\|e_{p,i-1}(t)\|_\infty,
\end{align}
where $\|e_{p,i}(t)\|_\infty$ is the $l_\infty$ norm of $e_{p,i}(t)$, i.e., $\|e_{p,i}(t)\|_\infty=\text{max}_{t\geq0}(|e_{p,i}(t)|)$. It means the maximum magnitude of the position error does not amplify along the platoon. 
\end{definition}

\section{Controller Design}\label{Control Algorithm Design}
\begin{figure*}[!t]
\includegraphics[width=0.9\linewidth]{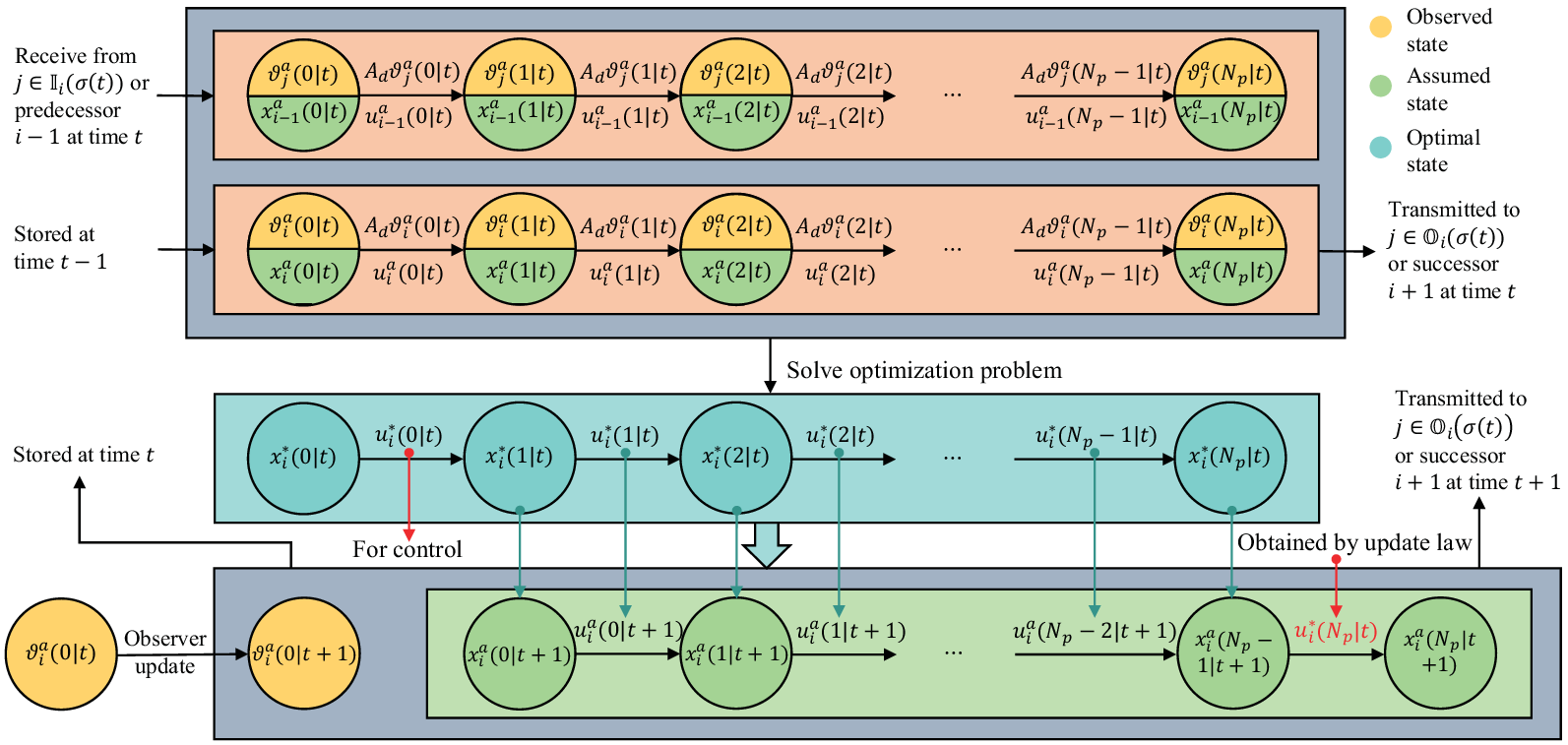}
\centering
\caption{Illustration of observer-based DMPC control strategy. Vehicle $i$ obtains the \textit{assumed trajectories of its predecessor} $x^{a}_{i-1}$, and the \textit{assumed observation trajectories of its in-neighbors vehicle} $j$ at time $t$, i.e., $\vartheta_j^a$, and then uses the information along with its own assumed trajectory $x^{a}_{i}$ and assumed observation trajectories $\vartheta_i^a$ to solve the optimization problem. The first control input of $u_i^*$ from the optimal solution is applied to control vehicle $i$ at time $t$, while the remaining control inputs are combined with the updated terminal control input to generate the assumed trajectory for time $t+1$.
}
\label{fig:principle_scheme}
\vspace{-0.6cm}
\end{figure*}
\subsection{Fully Distributed Adaptive Observer }\label{AA}
\par For every follower $i\in\mathbb{N}$, a fully distributed adaptive observer is designed as follows to directly get the leader's state or indirectly obtain this information from its  neighbors via V2V communication:
\begin{align}\label{eq:observer}
\begin{cases}\dot{\vartheta}_i(t)=A\vartheta_i(t)-\big(\varsigma_i(t)+\varrho_i(t)\big)\Psi_i(\varsigma_{i}(t))P\phi_i(t),\\\dot{\varrho}_i(t)=\phi_i^T(t)\Gamma\phi_i(t),\\\varsigma_i(t)=\phi_i^T(t)\Upsilon\phi_i(t),\end{cases}
\end{align}
where $\vartheta_{i}(t)\in\mathbb{R}^3$ is the state of observer $i$.  $ \varsigma_i(t), \varrho_i(t)\in\mathbb{R}^+$ indicate the adaptive coupling gains with $\varrho_{i}(0)\geq1$ and $\Psi_i(\cdot)$ is a nonlinear increasing non-negative function fulfilling $\Psi_i(\cdot)\geq1$ which will be designed later. Matrices $P$, $\Gamma$, and  $\Upsilon$ with the appropriate dimensions need to be designed. We define the observation error  $\theta_{i}(t) = \vartheta_{i}(t)- x_0 (t)=[\theta_{p,i},\theta_{v,i},\theta_{a,i}]^T$ and $x_0 (t)$ is the leader's state.  $\phi_i (t)$ is the relative observation error, defined as follows:
\begin{align}\label{eq:relative observation error}
\phi_{i}(t)=\sum_{j=1}^{N}m_{ij}(\sigma(t))\theta_{j}(t), 
\end{align}
where $m_{ij}(\sigma(t))$ is the element of $\mathcal{M}(\sigma(t))$. Denote $\kappa_i(t) = (\varsigma_i(t) + \varrho_i(t))\Psi_i(\varsigma_i(t))$. The update law of $\phi_{i}(t)$ gives
\begin{align}\label{eq:update of relative observation error}
\dot{\phi}_{i}(t)=&A\phi_{i}(t)-
\sum_{j=1}^{N}m_{ij}\kappa_j(t)P\phi_{j}(t).
\end{align}

\begin{theorem}
\label{thm: convergence of observer}
    Select coefficient matrix  $\Gamma=I_{n_0}$ and $\Upsilon=P^{-1}$ satisfying $\Gamma=\Upsilon P$, under Assumptions \ref{ass:continuous-time Markov process} and \ref{ass: a directed spanning tree rooted at the leader}, the fully distributed observer  \eqref{eq:observer}-\eqref{eq:update of relative observation error} can reach consensus in the mean square sense, i.e., $\lim_{t\to+\infty}\mathbb{E}[\theta_i(t)]=0$ and $\lim_{t\to+\infty}\mathbb{E}[\varrho_{i}(t)]=0$ for $i\in\{1,2,\cdotp\cdotp\cdotp,N\}$, if there exists a  positive-define solution $P$ that satisfies
\begin{align}\label{eq:theorem}
PA+A^{T}P-2P^{2}+\mathcal{Q}=0,
\end{align}
where the matrix $\mathcal{Q}$ is positive definite and $I_{n_0}$ stands for the $n_0\times n_0$ real identity matrix. $n_0$ is the appropriate dimension. \par For simplicity,  time $t$ and  Markov process $\sigma(t)$ will be omitted unless causing ambiguity.
\end{theorem}
\begin{proof}Consider the following candidate Lyapunov function for system \eqref{eq:observer}:
\begin{align}\label{eq:11}
    V=&\sum_{\varphi=1}^{\iota}V_{\varphi}=\sum_{\varphi=1}^{\iota}\mathbb{E}\left[{\left({\sum_{i=1}^{N}\frac{h_i}{2}(2\varrho_{i}+\varsigma_{i})\int^{\varsigma_{i}}_{0}\Psi_i(s)ds}\right.}\right.\nonumber\\
&\left.{\left.{\qquad\qquad\,+\frac{\lambda_{0}\pi_{\text{min}}}{4}\sum^{N}_{i=1}\tilde{\varrho}_i^2}\right)\times\textbf{1}_{\{\sigma(t)=\varphi\}}}\right],
\end{align}
where $\textbf{1}_{\{\sigma(t)=\varphi\}}$ stands for the Dirac measure over the set $\{\sigma(t)=\varphi\}$. $h_i$ is a positive constant and $\hat{h}=\left[h_{1},h_{2},\cdots ,h_{N}\right]^T$ is a left-eigenvector associated with eigenvalue of matrix $\bar{\mathcal{M}}$. $\bar{\mathcal{M}}$ corresponds to the union graph $\bar{\mathcal{G}}$. Define $H\triangleq \text{diag}\left(h_{1},h_{2},\cdots,h_{N}\right)$, $\lambda_{0}\triangleq\frac{\lambda_{2}\left(H\bar{\mathcal{M}}+\bar{\mathcal{M}}^{T}H\right)}{N}$, $\lambda_{2}(\cdot)$ represents the smallest nonzero eigenvalue of the matrix. $ \pi_{\text{min}}\triangleq \text{min}^{\iota}_{\varphi=1}\pi_\varphi$ and $\tilde{\varrho_i}=\varrho_i-\alpha$, where parameter $\alpha$ is a positive constant to be determined later. The function $V_\varphi$ is positive definite, and it follows from Lemma 3.6 in \cite{do2012continuous} that 
\begin{align}\label{eq:14}
\nonumber
dV_{\varphi}=&\mathbb{E}\left[\sum^{N}_{i=1}\frac{h_i}{2}\left(2\dot{\varrho}_{i}+\dot{\varsigma}_{i}\right)\int^{\varsigma_{i}}_{0}\Psi_i(s)ds \times\textbf{1}_{\{\sigma(t)=\varphi\}}\right]dt\\\nonumber&+\mathbb{E}\left[\sum^{N}_{i=1}\frac{h_i}{2}\left(2\varrho_{i}+\varsigma_{i}\right)\Psi_i\dot{\varsigma}_{i}  \times\textbf{1}_{\{\sigma(t)=\varphi\}}\right]dt\\\nonumber&+\mathbb{E}\left[\frac{\lambda_0\pi_{\text{min}}}{2}\sum^{N}_{i=1}\tilde{\varrho}_i\dot{\varrho}_i \times\textbf{1}_{\{\sigma(t)=\varphi\}}\right]dt\\&+\sum^\iota_{q=1}\mu_{q\varphi}V_{q}dt+o(dt).
\end{align}
\par Considering the fact that $\mu\cdot\textbf{1}=\textbf{0}$, the derivative of $V$ can be expressed as $\dot{V}=\sum^\iota_{\varphi=1}\dot{V}_\varphi$.
\par For $i\in\mathbb{N}$, the monotonically increasing functions $\Psi_i$ fulfill $\Psi_i(s)\geq1$ and $s>0$, using the mean value theorem for integrals, one can obtain that
\begin{align}\label{eq:15}
\nonumber
&\quad\mathbb{E}\left[\sum^{N}_{i=1}\frac{h_i}{2}\left(2\dot{\varrho}_{i}+\dot{\varsigma}_{i}\right)\int^{\varsigma_{i}}_{0}\Psi_i(s)ds \,\textbf{1}_{\{\sigma(t)=\varphi\}}\right]\\&\leq\mathbb{E}\left[\sum^N_{i=1}h_i\varsigma_i\Psi_i\phi_i^T\Gamma\phi_i\,\textbf{1}_{\{\sigma(t)=\varphi\}}\right]\nonumber\\&\quad
+\mathbb{E}\left[\sum^N_{i=1}\frac{h_i}{2}\dot{\varsigma_i}\varsigma_i\Psi_i\,\textbf{1}_{\{\sigma(t)=\varphi\}}\right].
\end{align}

  \par Using \eqref{eq:15} and \eqref{eq:update of relative observation error} yields
    \begin{align}\label{eq:16}
  \nonumber
\dot{V}&\leq\sum^\iota_{\varphi=1}\mathbb{E}\left[\sum^N_{i=1}h_i\varsigma_i\Psi_i\phi_i^T\Gamma\phi_i\,\textbf{1}_{\{\sigma(t)=\varphi\}}\right]\nonumber
\\&\quad+\sum^\iota_{\varphi=1}\mathbb{E}\left[\sum^N_{i=1}2h_i(\varrho_i+\varsigma_i)\Psi_i\phi_i^T\Upsilon A\phi_i\,\textbf{1}_{\{\sigma(t)=\varphi\}}\right]\nonumber
\\&\quad-\sum^\iota_{\varphi=1}\mathbb{E}\Bigg[2h_i(\varrho_i+\varsigma_i)\Psi_i\phi_i^T\Upsilon\nonumber
\\&\quad\times\bigg(\sum^N_{j=1}m_{ij}\kappa_j(t)P\phi_j\bigg)\,\textbf{1}_{\{\sigma(t)=\varphi\}}\Bigg]\nonumber
\\&\quad+\sum^\iota_{\varphi=1}\mathbb{E}\left[\frac{\lambda_0\pi_\text{min}}{2}\sum^N_{i=1}\tilde{\varrho}_i\phi_i^T\Gamma\phi_i\,\textbf{1}_{\{\sigma(t)=\varphi\}}\right].
  \end{align}
\par Given that $\sigma(t)$ has a unique distribution equal to $\pi$, by Lemma 1 from \cite{wang2022distributed}, it follows that 
\begin{align}\label{eq:17}
\nonumber
&\sum^\iota_{\varphi=1}\mathbb{E}\Bigg[2h_i(\varrho_i+\varsigma_i)\Psi_i\phi_i^T\Upsilon \sum^N_{j=1}m_{ij}\kappa_j(t)P\phi_j\,\textbf{1}_{\{\sigma(t)=\varphi\}}\Bigg]\nonumber
 \\\geq&\lambda_0\pi_{\text{min}}\mathbb{E}\left[\Phi^T\left(\Omega^2G^2\right)\otimes\left(\Upsilon P\right)\Phi\right],
\end{align}
 where diagonal matrix $\Omega=\text{diag}\left(\varrho_1+\varsigma_1,\cdots,\varrho_N+\varsigma_N\right)$ and $G=\text{diag}\left(\Psi_1,\cdots,\Psi_N\right)$, and $\Phi=[\phi_1,\cdots,\phi_N]^T$. $\otimes$ means Kronecker product.
If $\Gamma=\Upsilon P$ is true, then substituting \eqref{eq:17} into \eqref{eq:16} and using the Young's inequality, one has
\begin{align}\label{eq:18}
  \nonumber
\dot{V}&\leq\mathbb{E}\Bigg[\sum^N_{i=1}h_i(\varrho_i+\varsigma_i)\Psi_i\phi_i^T\left(\Upsilon A+A^T\Upsilon\right)\phi_i\Bigg]\nonumber
\\&\quad-\lambda_0\pi_{\text{min}}\mathbb{E}\left[\sum^N_{i=1}\left(\left(\varrho+\varsigma\right)^2\Psi_i^2-\frac{\varrho_i}{2}-\frac{\varsigma_i^2}{2}\Psi_i^2\right)\phi_i^T\Gamma\phi_i\right]\nonumber
\\&\quad-\mathbb{E}\Bigg[\sum^N_{i=1}\left(\frac{\lambda_0\pi_{\text{min}}}{2}\alpha-\frac{h_i^2}{2\lambda_0\pi_{\text{min}}}\right)\phi_i^T\Gamma\phi_i\Bigg].
 \end{align}
\par Noting that $\Psi_i\geq1$ and $\varrho_{i}\geq1$, one has
\begin{align}\label{eq:19}
\frac{\varrho_i}{2}+\frac{\varsigma_i^2}{2}\Psi_i^2\leq(\frac{\varrho_i^2}{2}+\frac{\varsigma_i^2}{2}+\varrho_i\varsigma_i)\Psi_i^2=\frac{1}{2}(\varrho_i+\varsigma_i)^2\Psi_i^2.
 \end{align}
 \par Let $\alpha\geq\hat{\alpha}+\text{max}_{i=1}^N\frac{h_i^2}{(\lambda_0\pi_{\text{min}})^2}$, $\hat{\alpha}$ will be determined later. Substituting \eqref{eq:19} into \eqref{eq:18}, it follows that
 \begin{align}\label{ea:20}
\dot{V}&\leq\mathbb{E}\Bigg[\sum^N_{i=1}h_i(\varrho_i+\varsigma_i)\Psi_i\phi_i^T\left(\Upsilon A+A^T\Upsilon\right)\phi_i\Bigg]\nonumber
\\&\quad-\frac{\lambda_0\pi_{\text{min}}}{2}\mathbb{E}\left[\sum^N_{i=1}\left(\left(\varrho_i+\varsigma_i\right)^2\Psi_i^2+\hat{\alpha}\right)\phi_i^T\Gamma\phi_i\right].
 \end{align}
 \par Using the Young's inequality and selecting $\hat{\alpha}$ satisfying $\sqrt{\hat{\alpha}}\lambda_0\pi_{\text{min}}\geq2\text{max}^N_{i=1}h_i$, one has
 \begin{align}\label{eq:21}
&\quad\frac{\lambda_0\pi_{\text{min}}}{2}\mathbb{E}\left[\sum^N_{i=1}\left(\left(\varrho_i+\varsigma_i\right)^2\Psi_i^2+\hat{\alpha}\right)\phi_i^T\Gamma\phi_i\right]\nonumber
\\&\geq\lambda_0\pi_{\text{min}}\mathbb{E}\left[\sum^N_{i=1}\sqrt{\hat{\alpha}}\left(\varrho_i+\varsigma_i\right)\Psi_i\phi_i^T\Gamma\phi_i\right]\nonumber
\\&\geq\mathbb{E}\left[2\sum^N_{i=1}h_i\left(\varrho_i+\varsigma_i\right)\Psi_i\phi_i^T\Gamma\phi_i\right].
 \end{align}
 \par Let $\bar{\Phi}=\left(\sqrt{\Omega GH}\otimes I_{n_0}\right)\Phi$, then substituting \eqref{eq:21} into \eqref{ea:20} gives  
 \begin{align}\label{eq:22}
\dot{V}\leq\mathbb{E}\left[\bar{\Phi}^T\left(I_N\otimes\left(\Upsilon A+A^T\Upsilon-2\Gamma\right)\right)\bar{\Phi}\right].
 \end{align}
\par Choose  $\Gamma=I_{n_0}$ and $\Upsilon=P^{-1}$ and consider that
\begin{align}\label{eq:23}
P^{-1}A+A^TP^{-1}-2I_{n_0}<0,
\end{align}
where the inequality \eqref{eq:theorem} is derived by multiplying both sides of \eqref{eq:23} by $P$, it follows that
\begin{align}\label{eq:24}
\dot{V}\leq-\lambda_{\text{min}}\mathbb{E}\left[\|\Phi\|^2\right]\leq0,
\end{align}
 where let $\lambda_{\text{min}}$ be the smallest eigenvalue of the matrix $-\left(\Omega GH\otimes\left(P^{-1}A+A^TP^{-1}-2I_{n_0}\right)\right)$. Hence, the Lyapunov function $V$ is bounded. 
 \par Assume that $\lim_{t\to+\infty}\mathbb{E}\left[\|\Phi\|^2\right]=\xi$ and $\xi>0$, there exists a time instant $t^*$ such that $\mathbb{E}\left[\|\Phi\|^2\right]>\frac{\xi}{2}$ when $t\geq t^*$.
 \par By integrating \eqref{eq:25} from $t^*=0$ to $t=+\infty$, we can obtain that
 \begin{align}\label{eq:25}
+\infty=\int^{+\infty}_{t^*}\frac{\xi}{2}dt&<\int^{+\infty}_{t^*}\mathbb{E}\left[\|\Phi\|^2\right]dt\nonumber\\&<\frac{1}{\lambda_{\text{min}}}\left(V\left(t^*\right)-V\left(+\infty\right)\right)=-\infty.
 \end{align}
 \par This provides a contradiction-based proof. Thus $\xi=0$, then we can obtain $\lim_{t\to+\infty}\mathbb{E}\left[\|\Phi\|^2\right]=0$. Furthermore, considering Jensen’s
inequality, one has $\mathbb{E}\left[\|\Phi^2\|\right]\geq\left(\mathbb{E}\left[\|\Phi\|\right]\right)^2$. Then we can derive $\lim_{t\to+\infty}\mathbb{E}\left[\|\Phi\|\right]=0$. Considering \eqref{eq:relative observation error}, it is obtained that $\lim_{t\to+\infty}\mathbb{E}\left[\theta_i\right]=0$ and $\lim_{t\to+\infty}\mathbb{E}\left[\varrho_i\right]=0$. Therefore, the observer states in \eqref{eq:observer} can achieve consensus in the mean square sense. The proof is completed.
\end{proof}
\begin{remark}
Compared with \cite{wang2024cyber}, quadratic extra gain $\varsigma_i(t)$ is introduced to assist the derivation to adapt to the directed communication topology \cite{liu2024time}. Inspired by \cite{wang2022distributed}, an adaptive coupling gain $\Psi_i(\varsigma_{i}(t))$ is introduced to accelerate the convergence of observer states.  We constructed the Lyapunov function \eqref{eq:11} using  $\varsigma_i(t)$ and $\Psi_i(\varsigma_{i}(t))$, and provided a sufficient condition \eqref{eq:theorem} in Theorem \ref{thm: convergence of observer}. Based on the proof of Theorem \ref{thm: convergence of observer}, we note that the selection of  $\Psi_i(\varsigma_{i}(t))$ offers flexibility and can be determined experimentally.
\end{remark}
\subsection{Distributed Model Predictive Control Algorithm}
\par To model the DMPC framework (see Fig.  \ref{fig:principle_scheme}), the Euler method is used to discretize system dynamics \eqref{eq:continuous linear follower model} and \eqref{eq:continuous leader model} with sampling time $\Delta t$ as follows, which is also used in \cite{wang2025string}:
 \begin{align}
&x_i(k+1)=A_dx_i(k)+B_du_i(k),\nonumber\\
     &A_d=\begin{bmatrix}1&\Delta t&0\\0&1&\Delta t\\0&0&1\end{bmatrix},B_d=\begin{bmatrix}0\\0\\\Delta t\end{bmatrix},\label{eq:discrete follower model}\\
     &x_0(k+1)=A_dx_0(k).\label{eq:discrete leader model}
 \end{align}
\par The control strategy for the DMPC framework is shown in Fig. \ref{fig:principle_scheme}. At time $t$, vehicle $i$ obtains the assumed trajectories of its predecessor and the assumed observation trajectories (generated by \eqref{eq:discrete leader model}) of its in-neighbors, and then uses this information, along with its own assumed trajectory and assumed observation trajectory, to solve the optimal control problem. The first control input from the optimal solution is applied to control vehicle $i$ at time $t$, while the remaining control inputs are combined with the updated terminal control input to generate the assumed trajectory for time $t+1$.  $N_p$ is the predictive horizon. For arbitrary integers $N_1$, $N_2$ satisfying $N_1<N_2$, define $\mathbb{K}_{N_1:N_2}=\{N_1,N_1+1,\cdots,N_2-1,N_2\}$. We denote $\mathbb{K}_{0:N_p-1}=\mathbb{K}_1$ and $\mathbb{K}_{0:N_p}=\mathbb{K}_2$. Three types of state trajectories and control input trajectories are defined as follows:
\par i) The predicted trajectory $x_i^p(k|t),k\in\mathbb{K}_{2}$ and $u_i^p(k|t),k\in\mathbb{K}_{1}$, the variables to be optimized; 
\par ii) The optimal trajectory $x_i^*(k|t),k\in\mathbb{K}_{2}$ and $u_i^*(k|t),k\in\mathbb{K}_{1}$, the optimal solution of the local optimization problem; 
\par iii) The assumed trajectory $x_i^a(k|t),k\in\mathbb{K}_{2}$ and $u_i^a(k|t),k\in\mathbb{K}_{1}$, the parameters generated with the optimal trajectories.
\par Additionally, we denote by $\vartheta_i^a(k|t),k\in\mathbb{K}_{2}$ the assumed observation trajectory. Based on the observation of the leader, we can design an optimization problem $\mathcal{P}_i(t)$ for vehicle $i$ at time $t$ as follows:
\begin{subequations}
\begin{align}
&\min_{u_i^p(k|t),k\in\mathbb{K}_{1}}J_i\left(x_i^p,u_i^p,x_i^a,\vartheta_i^a,\vartheta_j^a,x_{i-1}^a;\mathbb{K}_{1}|t\right)\nonumber\\
&=\sum_{k=0}^{N_p-1}l_i\left(x_i^p,u_i^p,x_i^a,\vartheta_i^a,\vartheta_j^a,x_{i-1}^a;k|t\right)\label{eq:28a}\\
\
&s.t.\qquad x_i^p(0|t)=x_i(t),\label{eq:28b}\\
&x_i^p(k+1|t)=A_dx_i^p(k|t)+B_du_i^p(k|t),k\in\mathbb{K}_{1},\label{eq:28c}\\
&\qquad u_i^p(k|t)\in\mathbb{U}_i,k\in\mathbb{K}_{1},\label{eq:28d}\\
&\qquad x_i^p(N_p|t)=x_i^a(N_p|t),\label{eq:28e}\\
&\qquad\left\|\gamma\hat{e}_i^p(k|t)\right\|_{\infty} \leq \beta_{i}D_{i-1}(t),k\in\mathbb{K}_{2},\label{eq:28f}
\end{align}
\end{subequations}
where $\mathbb{U}_i$ is the control input set. \eqref{eq:28b}, \eqref{eq:28c}, and \eqref{eq:28d} represent the constraint of the initial state, vehicle dynamics, and control input, respectively.  \eqref{eq:28e} is the terminal state constraint and \eqref{eq:28f} is the observer-based string stable constraint.
\par In constraint \eqref{eq:28f},  $\gamma=[1,0,0]$. To obtain more accurate observations of the leader, we define the average observation as $\vartheta_{i,avg}(t)=\frac{1}{|\mathbb{I}_i(\sigma(t))|+1}\left[\vartheta_i(t)+\sum_{j\in\mathbb{I}_i(\sigma(t))}\vartheta_j(t)\right]$, 
where $|\mathbb{I}_i(\sigma(t))|$ is cardinality of $\mathbb{I}_i(\sigma(t))$. We denote  $\hat{e}_i(t)=x_i(t)-\vartheta_{i,avg}(t)+\tilde{d}_{i0}$. Then we define $D_{i-1}=\text{max}\left\{\|\gamma \hat{e}_{i-1}(\tau)\|_\infty, \|\gamma \hat{e}^a_{i-1}(k|t)\|_\infty\right\}$, 
where $\|\gamma \hat{e}_{i-1}(\tau)\|_\infty=\text{max}_{0<\tau\leq t}|\gamma \hat{e}_{i-1}(\tau)|$ represents the max spacing error of preceding vehicle $i-1$ in the historical trajectory and $\|\gamma \hat{e}_{i-1}^a(k|t)\|_\infty=\text{max}_{k\in\mathbb{K}_2}|\gamma \hat{e}_{i-1}^a(k|t)|$ represents the max spacing error of preceding vehicle $i-1$ in the assumed trajectory. The cost function is defined as:
\begin{align}\label{eq:cost}
  & l_i(x_i^p,u_i^p,x_i^a,\vartheta_i^a,\vartheta_j^a,x_{i-1}^a;k|t)\nonumber\\
=&\left\|u_i^p(k|t)\right\|_{R_i}+\left\|x_i^p(k|t)-x_i^a(k|t)\right\|_{F_i}\nonumber\\
  & +\left\|x_i^p(k|t)-x_{i-1}^a(k|t)+\tilde{d}_0\right\|_{S_i}\nonumber\\
  &+\left\|x_i^p(k|t)-\vartheta_{i,avg}^a(k|t)+\tilde{d}_{i0}\right\|_{G_i},
\end{align}
where $R_i$, $F_i$, $S_i$, and $G_i$ are positive definite diagonal matrices, representing the weight of the corresponding penalty term. $\left\|u_i^p(k|t)\right\|_{R_i}$ is the penalty for the control input; $\left\|x_i^p(k|t)-x_i^a(k|t)\right\|_{F_i}$ is the penalty for the deviation between the predicted and assumed states of vehicle $i$; $\left\|x_i^p(k|t)-x_{i-1}^a(k|t)+\tilde{d}_0\right\|_{S_i}$ is the penalty for the deviation between vehicle $i$'s predicted states and the preceding vehicle $i-1$'s assumed states with offset $\tilde{d}_0=[d_0,0,0]^T$. When $i=1$, the preceding vehicle is the leader; $\left\|x_i^p(k|t)-\vartheta_{i,avg}^a(k|t)+\tilde{d}_{i0}\right\|_{G_i}$ is the penalty for the deviation between vehicle $i$'s predicted states and average observation states with offset $\tilde{d}_{i0}$. Here we denote by $\vartheta_{i,avg}^a(k|t),k\in\mathbb{K}_{2}$ the average assumed observation trajectory.
\par From \eqref{eq:28f}, we derive $x_i^*(N_p|t)=x_i^a(N_p|t)$. The update law is formulated as:
\begin{align}\label{eq:terminal state update}
x_i^a(N_p|t+1)=A_dx_i^a(N_p|t)+B_du^*_i(N_p|t),
\end{align}
where $x_i^a(N_p|t+1)$ represents the assumed state of $N_p$ at time $t+1$ looking forward and $u_i^*(N_p|t)=K\cdotp\left(\vartheta_{i,avg}^a(N_p|t)-x_i^a(N_p|t)-\tilde{d}_{i0}\right)$. We note that $u_i^*(N_p|t)$ is defined only for notational simplicity, and the control gain matrix $K$ can be determined by the modified algebraic Riccati equation in \cite{bian2020distributed}.

\begin{remark}   
The transmission of the assumed trajectory is based on the PF (predecessor-following) communication topology, and the switching topology may lead to the loss of predecessor $i-1$'s information,  in which the vehicles can rely on additional observer information to ensure safety.
\end{remark}

\begin{remark}
If the conditions of Theorem \ref{thm: convergence of observer} are satisfied, the state of the observer of each follower converges to the leader's state. Then, the terminal constraints \eqref{eq:28e} and the modified update law \eqref{eq:terminal state update} guarantee recursive feasibility and terminal consensus \cite{bian2020distributed}. Furthermore, the modified string stability constraints ensure that the maximum predicted spacing error will not exceed the maximum spacing error of the predecessor \cite{wang2023dmpc}. Using the control algorithm in Fig. \ref{fig:principle_scheme}, the control objectives \eqref{eq:Mean-square closed-loop stability} and \eqref{eq:Predecessor-follower string stability} can be achieved under Markovian switching topology.
\begin{remark}
Due to the switching communication topology, the adjacent nodes in the topology are constantly changing, resulting in the changing information that each vehicle receives from other vehicles. We use an observer for each follower to obtain the average observation of the leader for DMPC control. This information is not affected by the time-varying communication topology. Therefore, our approach ensures robustness under dynamic communication topologies.
\end{remark}
\end{remark}
\begin{figure}[b]
\vspace{-0.5cm}
\includegraphics[width=0.85\linewidth]{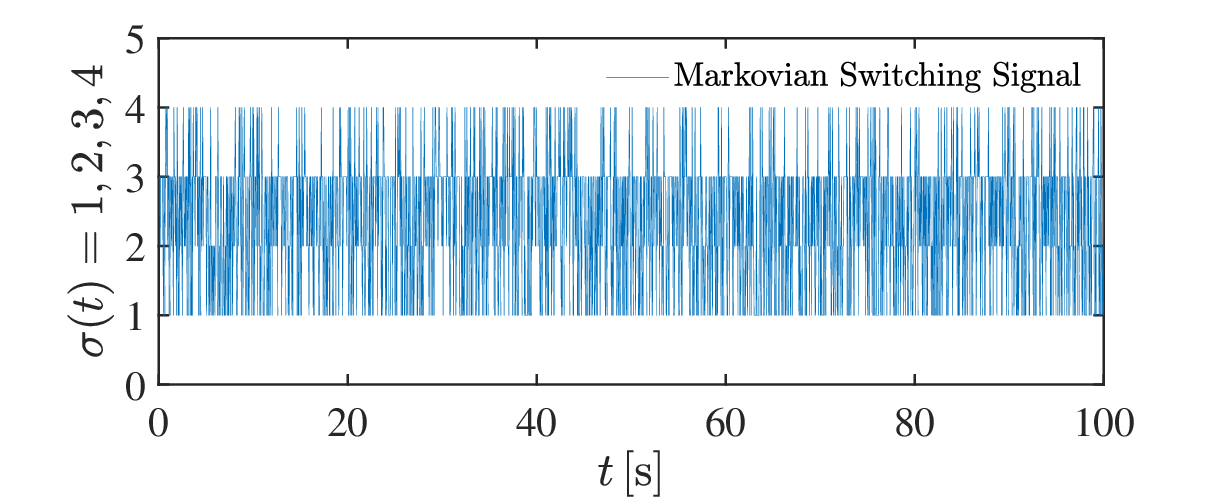}
\centering
\caption{Diagram of Markovian switching communication topologies. There are four types of communication topologies, i.e., $\sigma(t)=1,2,3,4$.}
\label{fig:sigma}
\end{figure}

\section{Numerical Experiments}\label{Numerical Experiments}
\subsection{Simulation Settings}
\par We consider a vehicle platoon with one leader and $N=5$ followers. The initial state of followers is given by $x_i(0)=[p_i(0),v_i(0),a_i(0)]^T=[i\cdotp d_0,0,0]^T$. The velocity of the leader can be defined as follows:
\begin{align}\label{desired velocity}
   v_0(t)=
   \begin{cases}
        t,\qquad& 0s\leq t<25s\\
        25,&25s\leq t<50s\\
        25-1.2(t-50),&50s\leq t<60s\\
        13,&60s\leq t<100s
    \end{cases}\; m/s.
\end{align}
\par In the simulation, we set the sampling time $\Delta t=0.1s$, the desired distance $d_0=20m$, and the control input bound $\mathbb{U}_i=[-3,3]m/s^2$. The predictive horizon is $N_p=10$. The weights are: $R_i=0.1$, $S_i=\text{diag}([5,2.5,1])$, $G_i=\text{diag}([50,25,10])$, $i=1,2,\cdots,5$, $F_i=\text{diag}([5,2.5,1])$, $i=1,2,3,4$, $F_5=\text{diag}([0,0,0])$. The control gain matrix $K=[1.66;5.39;2.42]$ and the scaling parameter $\beta_i=0.6$. For fairness, the benchmark controller in \cite{bian2020distributed} was evaluated using the same simulation parameters. 

\par We assume the communication topology switches with four possible topologies: $\mathcal{G}_1$: LPF (leader-predecessor following), $\mathcal{G}_2$: LPF-failure, $\mathcal{G}_3$: PF, $\mathcal{G}_4$: PF-failure, also used in \cite{li2020distributed,wang2024cyber}. In the PF-failure mode, the communication channel between vehicle 2 and 3 is broken, while in the PLF-failure, the communication channel between the leader and the vehicle 4 and 5 is broken, leading to system instability.
The union graph $\bar{\mathcal{G}}=\left\{\mathcal{G}_1,\mathcal{G}_2,\mathcal{G}_3,\mathcal{G}_4\right\}$ satisfies Assumption  \ref{ass:continuous-time Markov process} and \ref{ass: a directed spanning tree rooted at the leader}. Refer to \cite{wang2024cyber}, the transition rate matrix is selected as $\mu=
\begin{bmatrix}
-2 & 0.8 & 0.8 & 0.4 \\
1.2 & -2.4 & 0.8 & 0.4 \\
0.4 & 0.4 & -1.2 & 0.4 \\
1.2 & 0.8 & 0.8 & -2.8
\end{bmatrix}$. The corresponding invariant distribution is $\pi=[11/40,1/5,2/5,1/8]$, which indicates that the switching process is rapid and smooth, with $\sigma(t)=\left\{1,2,3,4\right\}$. According to the condition in Theorem \ref{thm: convergence of observer}, the matrix  $P=\begin{bmatrix}
1.5602&0.2230&0.0159\\
0.2230&1.6081&0.2275\\
0.0159&0.2275&1.6246
\end{bmatrix}$ and we choose function $\Psi_i(\varsigma_i(t))=(1+\varsigma_i(t))^\frac{1}{4}$ as the adaptive coupling gain.
To enable quantitative comparison, the measures of effectiveness (MOE) \cite{wang2023dmpc} is employed. Specifically, for all followers in the platoon, the max position error (MPE) and the max velocity error (MVE) quantify the maximum absolute errors in position and velocity; the average position error (APE) and the average velocity error (AVE) represent the average absolute error in position and velocity.


\subsection{Simulation Results}
\par Fig.~\ref{fig:sigma} depicts the switching behavior among four communication topologies. The observation error profiles between the leader and each observer are illustrated in Fig. \ref{fig:observer_mainfig}(a)-(c). The sharp changes in the curve are caused by variations in the velocity and acceleration of the leader vehicle. The observation errors $\theta_{p,i}$, $\theta_{v,i}$, and $\theta_{a,i}$ almost get close to zero. The adaptive parameter $\|\kappa_i(t)\|$ converges to finite positive values, as shown in Fig. \ref{fig:observer_mainfig}(d). This demonstrates the effectiveness and robustness against dynamic switching topologies. Fig. \ref{fig:tracking_mainfig} shows the tracking profiles of the proposed controller. From Fig.  \ref{fig:tracking_mainfig}(a), the followers’ position profile can track the leader’s position profile, and their trajectories do not exhibit any intersections, indicating no collisions between vehicles. Fig. \ref{fig:tracking_mainfig}(b) shows that the velocity profile can quickly track the desired velocity profile despite disturbances from the leader and switching communication topology.
\begin{figure}[!t]
    \centering
    \begin{subfigure}{0.48\linewidth}
        \includegraphics[width=\linewidth]{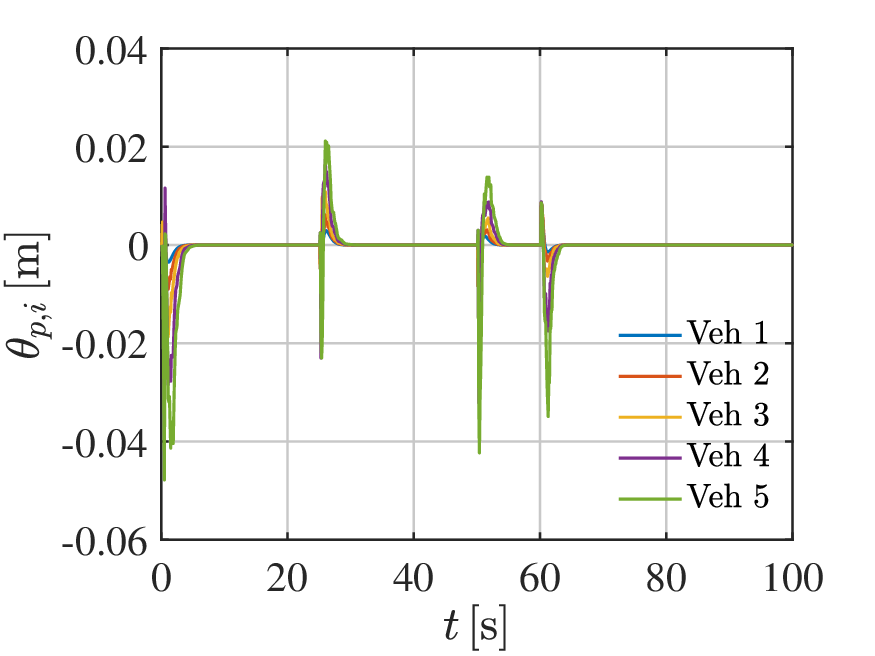}
        \caption{Position observation error.}
        \label{fig:subfig_p_tilde}
    \end{subfigure}
    \hfill
    \begin{subfigure}{0.48\linewidth}
        \includegraphics[width=\linewidth]{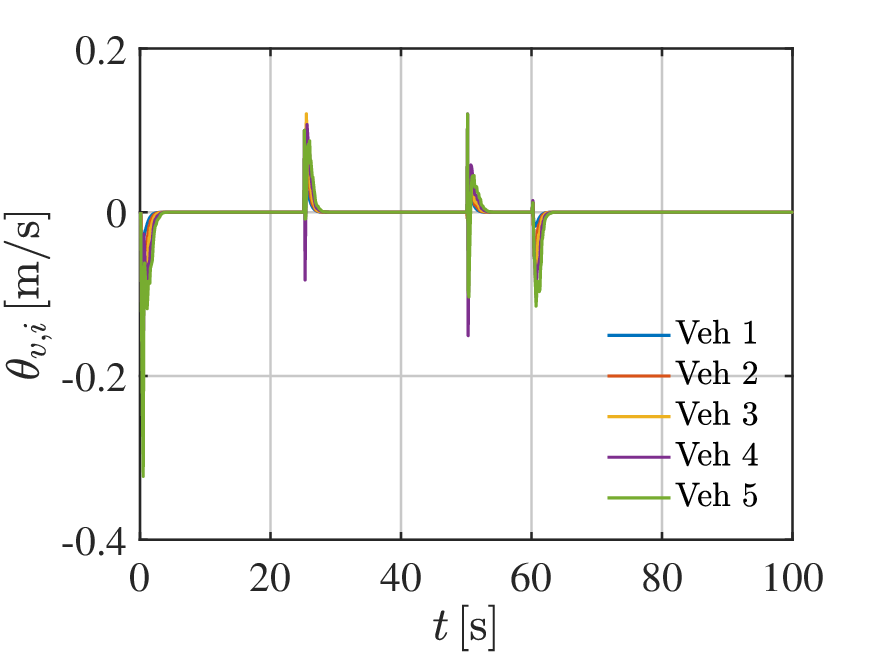}
        \caption{Velocity observation error.}
        \label{fig:subfig_v_tilde}
    \end{subfigure}
    \vfill
    \begin{subfigure}{0.48\linewidth}
        \includegraphics[width=\linewidth]{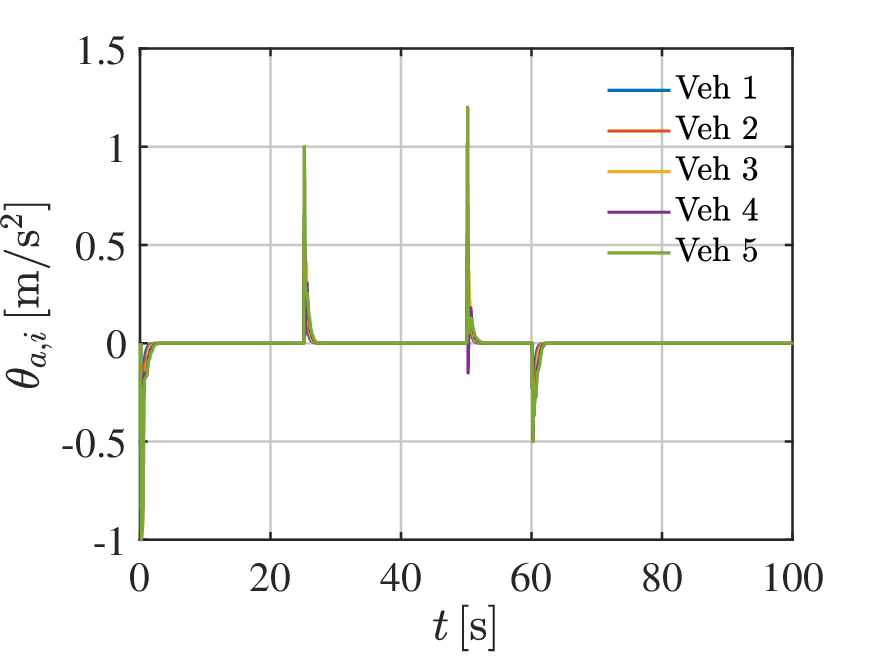}
        \caption{Acceleration observation error.}
        \label{fig:subfig_a_tilde}
    \end{subfigure}
    \hfill
    \begin{subfigure}{0.48\linewidth}
        \includegraphics[width=\linewidth]{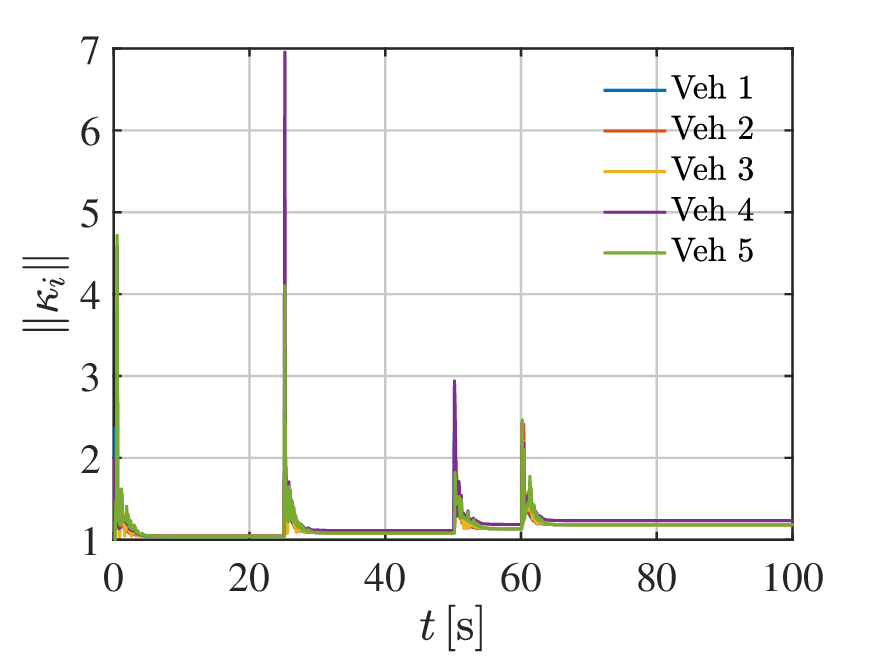}
        \caption{Adaptive parameter.}
        \label{fig:subfig_kappa}
    \end{subfigure}
    \caption{Profiles of observation errors between the leader and observers with Markovian switching topology and the evolution of the adaptive parameter.}
    \label{fig:observer_mainfig}
    \vspace{-0.5cm}
\end{figure}

\begin{figure}[!t]
    \centering
    \begin{subfigure}{0.48\linewidth}
        \includegraphics[width=\linewidth]{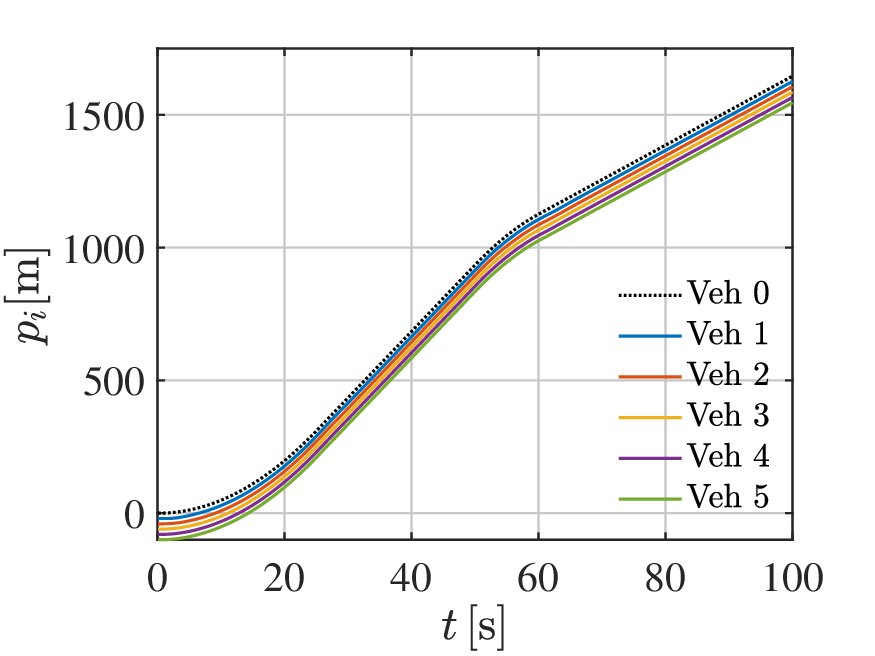}
        \caption{Position.}
        \label{fig:subfig_p}
    \end{subfigure}
    \hfill
    \begin{subfigure}{0.48\linewidth}
        \includegraphics[width=\linewidth]{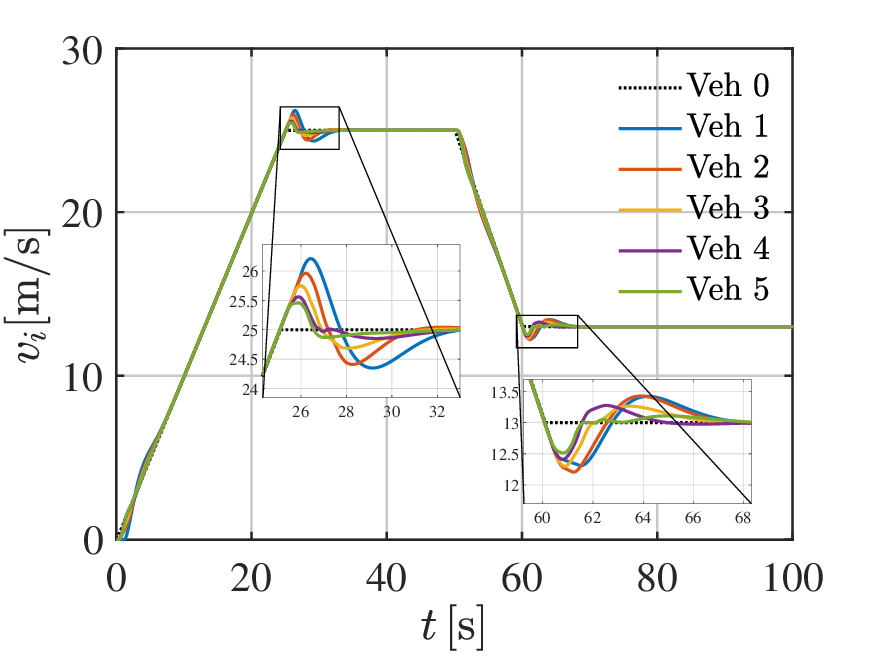}
        \caption{Velocity.}
        \label{fig:subfig_v}
    \end{subfigure}
    \caption{Tracking profiles of the proposed controller.}
    \label{fig:tracking_mainfig}
    \vspace{-0.8cm}
\end{figure}

\begin{table}[!t]
\vspace{0.5cm}
\renewcommand{\arraystretch}{1.5}
\centering
\caption{Performance Metrics of Controller.} 
\label{tab: The MOE of controller} 
\begin{tabular}{ccccc} 
\toprule 
\ &MPE [m]&MVE [m/s]&APE [m]&AVE [m/s]\\
\midrule
Our method&\textbf{1.83}&\textbf{1.21}&\textbf{0.13}&\textbf{0.07}\\
Method in \cite{bian2020distributed}&2.62&1.62&0.29&0.18\\
\bottomrule 
\end{tabular}
\vspace{-0.5cm}
\end{table}
The evaluation metrics are shown in Table \ref{tab: The MOE of controller},  and our method has smaller errors than \cite{bian2020distributed}, demonstrating the competitiveness of the proposed controller. The tracking error and control input are shown in Fig. \ref{fig:controller_mainfig}. As shown in Fig. \ref{fig:controller_mainfig} (a), (c), and (e), the proposed controller enables all followers to effectively track the leader, resulting in significantly smaller position, velocity, and acceleration errors than the benchmark controller (see Fig. \ref{fig:controller_mainfig} (b), (d), and (f)).  Additionally, the tracking errors of the proposed controller converge faster. In contrast, there are more frequent oscillations in the control inputs of the benchmark controller due to switching topologies, as shown in Fig.  \ref{fig:controller_mainfig} (h). The control input response of our controller in Fig.  \ref{fig:controller_mainfig} (g) is faster, with fewer instances of jitter. Moreover, the benchmark controller does not satisfy the string stability, see Fig. \ref{fig:controller_mainfig} (b). Our method in Fig. \ref{fig:controller_mainfig} (a) shows that the peak magnitude of the position error is not amplified along the platoon, showing predecessor-follower string stability.

\begin{figure}[t]
    \centering
    \begin{subfigure}{0.48\linewidth}
        \includegraphics[width=\linewidth]{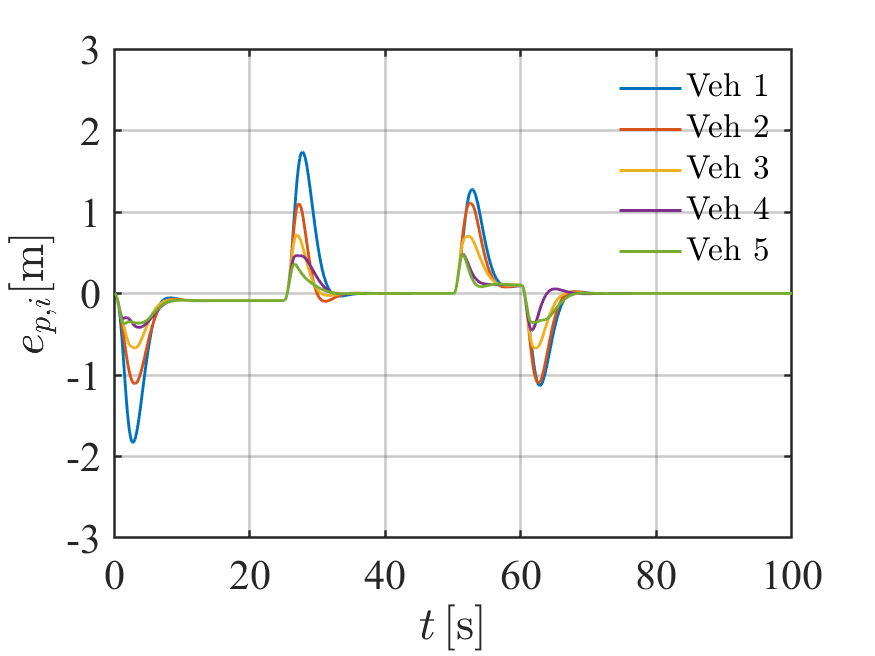}
        \caption{Position error (ours).}
        \label{fig:subfig_ep_proposed}
    \end{subfigure}
    \hfill
    \begin{subfigure}{0.48\linewidth}
        \includegraphics[width=\linewidth]{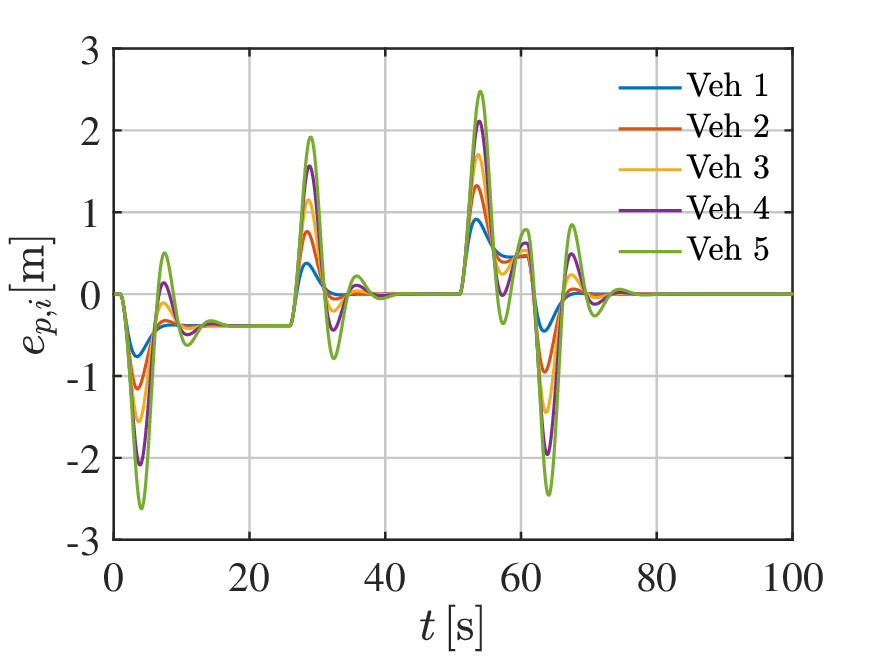}
        \caption{Position error (method in \cite{bian2020distributed}).}
        \label{fig:subfig_ep_benchmark}
    \end{subfigure}
    \vfill
    \begin{subfigure}{0.48\linewidth}
        \includegraphics[width=\linewidth]{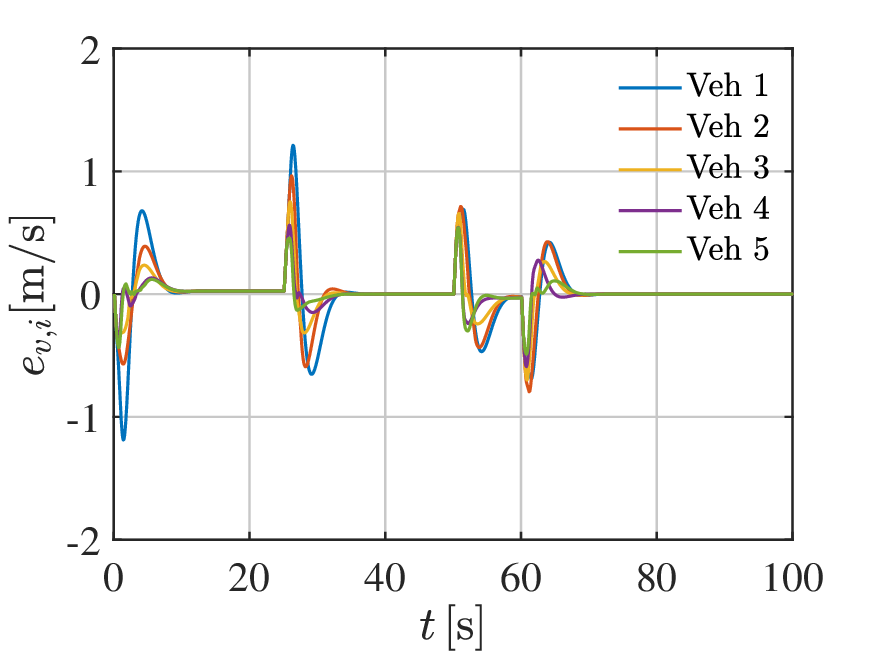}
        \caption{Velocity error (ours).}
        \label{fig:subfig_ev_proposed}
    \end{subfigure}
    \hfill
    \begin{subfigure}{0.48\linewidth}
        \includegraphics[width=\linewidth]{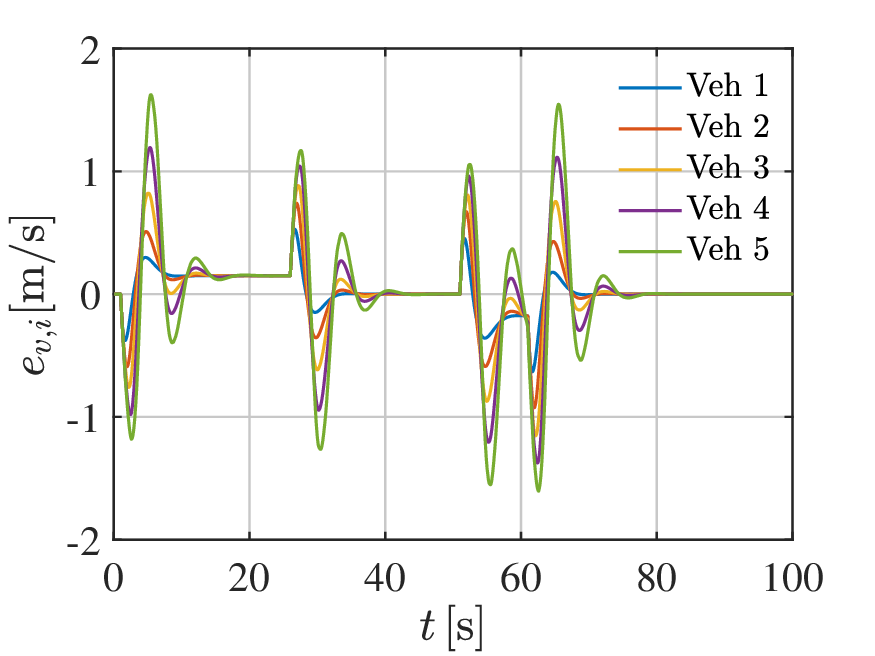}
        \caption{Velocity error (method in \cite{bian2020distributed}).}
        \label{fig:subfig_ev_benchmark}
    \end{subfigure}
    \vfill
    \begin{subfigure}{0.48\linewidth}
        \includegraphics[width=\linewidth]{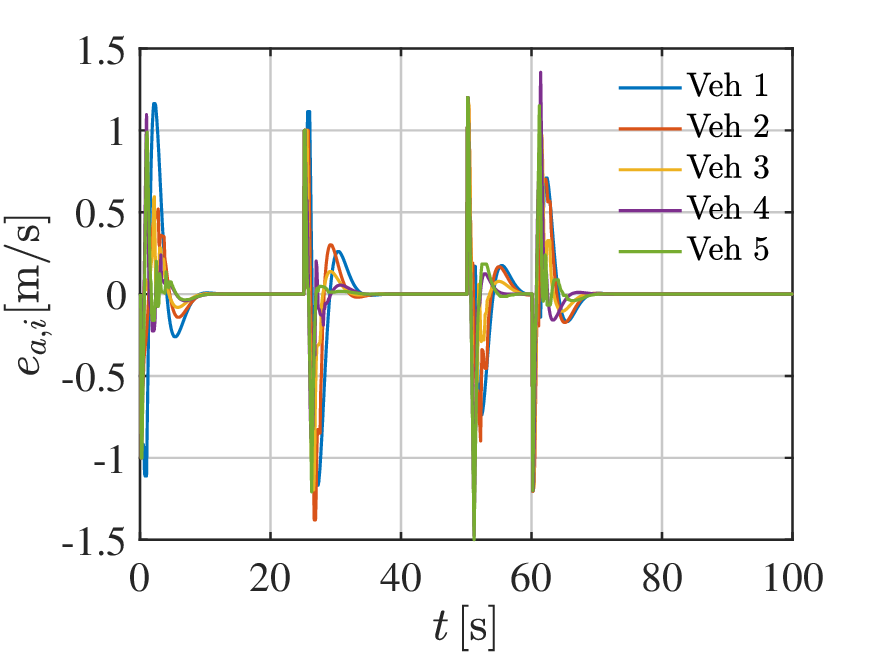}
        \caption{Acceleration error (ours).}
        \label{fig:subfig_ea_proposed}
    \end{subfigure}
    \hfill
    \begin{subfigure}{0.48\linewidth}
        \includegraphics[width=\linewidth]{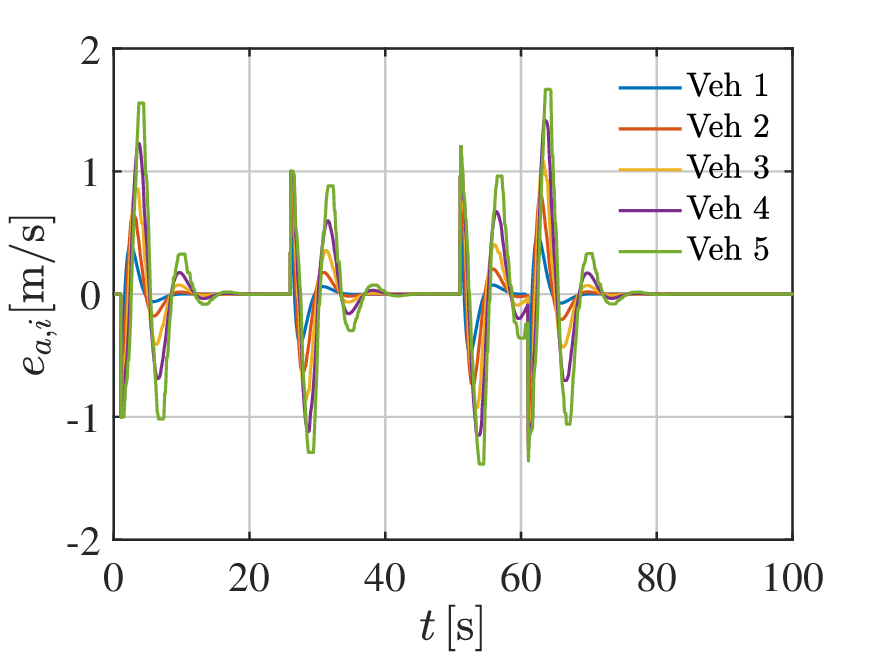}
        \caption{Acceleration error (method in \cite{bian2020distributed}).}
        \label{fig:subfig_ea_benchmark}
    \end{subfigure}
    \vfill
    \begin{subfigure}{0.48\linewidth}
        \includegraphics[width=\linewidth]{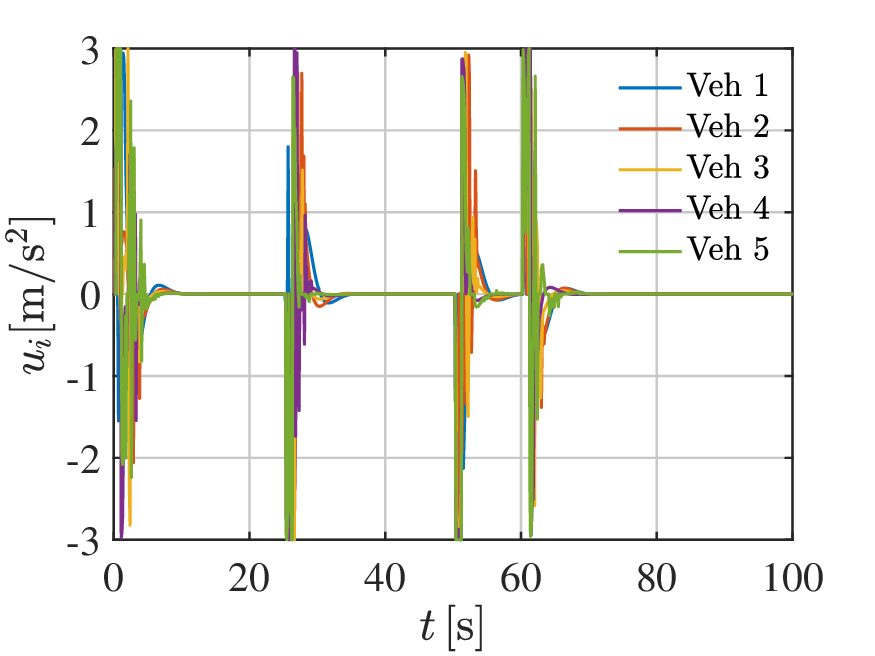}
        \caption{Control input (ours).}
        \label{fig:subfig_u_proposed}
    \end{subfigure}
    \hfill
    \begin{subfigure}{0.48\linewidth}
        \includegraphics[width=\linewidth]{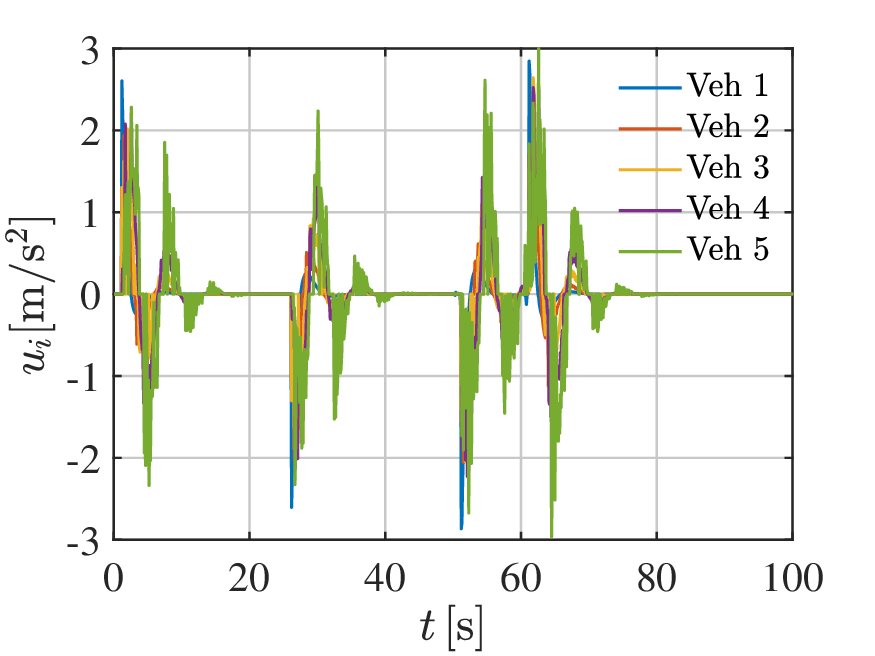}
        \caption{Control input (method in \cite{bian2020distributed}).}
        \label{fig:subfig_u_benchmark}
    \end{subfigure}
    \caption{Profiles of tracking errors: ours v.s. benchmark controller \cite{bian2020distributed}.}
    \label{fig:controller_mainfig}
    \vspace{-0.6cm}
\end{figure}
\section{Conclusion}\label{Conclusion}
\par This study proposes an observer-based DMPC framework to ensure stability under directed dynamic switching topologies. We use Markovian switching topologies and design a fully distributed adaptive observer to obtain the leader's information, unaffected by the switching of topologies. A sufficient condition for the asymptotic mean-square stability of the observer error system is provided. The terminal update law of the DMPC is constructed to guarantee asymptotic terminal consensus. Moreover, string stability constraints are established based on the observed information, ensuring predecessor-follower string stability. Future work will focus on extending this approach to mixed platoon control and examining the effects of communication delays.



\bibliographystyle{IEEEtran}
\small\bibliography{IEEEabrv,reference}

\end{document}